\documentclass{amsart}
\usepackage[a4paper,margin=25mm]{geometry}
\usepackage{amsmath}
\usepackage{amssymb}

\usepackage{tikz,hyperref}

\let\set\mathbb
\def\<#1>{\langle#1\rangle}

\def\eatspace#1{#1}

\def\step#1#2{%
  \par\kern1pt\dimen144=#2em\advance\dimen144by1.67em
  \hangindent=\dimen144\hangafter=1
  \leavevmode\rlap{\small#1}\kern\dimen144\relax\eatspace}

\newcommand{\N}{\mathbb{N}}
\newcommand{\T}{\mathcal{T}}
\newcommand{\Mnmp}{\mathcal{M}_{n,m,p}}
\newcommand{\Z}{\mathbb{Z}}
\newcommand{\rank}{\operatorname{rank}}

\renewcommand{\Gamma}{\varGamma}
\newcommand{\A}{\mathbf{A}}
\newcommand{\B}{\mathbf{B}}
\def\C{\mathbf{C}}
\newcommand{\range}[1]{\{1,\ldots,#1\}}

\renewcommand{\t}[1]{A^{(#1)} \otimes B^{(#1)} \otimes \Gamma^{(#1)}}

\newcommand{\bc}[1]{B^{(#1)} \otimes \Gamma^{(#1)}}

\newtheorem{theorem}{Theorem}
\newtheorem{proposition}[theorem]{Proposition}
\newtheorem{lemma}[theorem]{Lemma}
\newtheorem{definition}[theorem]{Definition}

\newtheorem{algorithm}{Algorithm}
\newtheorem{corollary}[theorem]{Corollary}
\newtheorem{question}{Question}

\begin{document}

 \title{Flip Graphs for Matrix Multiplication}

%
%

\author[Manuel Kauers]{Manuel Kauers\,$^\ast$}
\author[Jakob Moosbauer]{Jakob Moosbauer\,$^{\dagger}$}
 \address{Manuel Kauers, Institute for Algebra, J. Kepler University Linz, Austria}
 \email{manuel.kauers@jku.at}
 \address{Jakob Moosbauer, Institute for Algebra, J. Kepler University Linz, Austria}
 \email{jakob.moosbauer@jku.at}
 \thanks{$^\ast$ Supported by the Austrian FWF grants P31571-N32 and I6130-N}
 \thanks{$^\dagger$ Supported by the Land Oberösterreich through the LIT-AI Lab}
 
 \begin{abstract}
   We introduce a new method for discovering matrix multiplication schemes
   based on random walks in a certain graph, which we call the flip graph.
   Using this method, we were able to reduce the number of multiplications
   for the matrix formats $(4,4,5)$ and $(5,5,5)$, both in characteristic
   two and for arbitrary ground fields.
 \end{abstract}



 
 \maketitle

 \section{Introduction}\label{sec:introduction}

 Nobody knows the computational cost of computing the product of two matrices.  Strassen's
 discovery~\cite{St:Gein} that two $2\times2$-matrices can be multiplied with only 7 multiplications
 in the ground field launched intensive research on the complexity of matrix multiplication during
 the past decades. One branch of this research aims at finding upper (or possibly also lower) bounds on the
 matrix multiplication exponent~$\omega$. The current world record $\omega<2.37188$ is held by
 Duan, Wu and Zhou~\cite{DWZ:FMMv} and only slightly improves the previous record $\omega< 2.37286$
 by Alman and Williams~\cite{AW:ARLM}. These results concern asymptotically large matrix sizes.
 
 Another branch of research on matrix multiplication algorithms concerns specific small matrix
 sizes. For $2\times2$ matrices, it is known that there is no way to do the job with only 6
 multiplications~\cite{Wi:Omo2}, and that Strassen's algorithm is essentially the only way to do it
 with~7~\cite{dG:Ovoo}. Also for multiplying a $2\times 2$ matrix with a $2\times p$ matrix and for
 multiplying a $2\times 3$ matrix with a $3\times 3$ matrix, optimal algorithms are
 known~\cite{HK:OMtN}.  For all other formats, the known upper and lower bounds do not match. For
 example, for the case $3\times 3$ times $3\times 3$, the best known upper bound is
 23~\cite{La:Anaf} and the best known lower bound is 19~\cite{Bl:Otco} unless we impose restrictions
 on the ground domain such as commutativity.

 An upper bound for a specific matrix format can be obtained by stating an explicit matrix
 multiplication scheme with as few multiplications as possible. Such schemes can be discovered by
 various techniques, including hand calculation~\cite{St:Gein,La:Anaf}, numerical
 methods~\cite{Sm:Tbca,SS:TTRo}, SAT solving~\cite{CBH:ANGM,heule19a,HKS:Nwtm}, or machine
 learning~\cite{FBH+:Dfmm}. The latter approach, due to Fawzi et al., has received a lot of attention, even in
 the general public, because it led to an unexpected improvement of the upper bound for multiplying
 two $4\times4$ matrices from 49 to 47 multiplications in characteristic two. They also reduced the
 bound for multiplying two $5\times 5$ matrices from 98 to 96 in characteristic two, and found improvements for
 some formats involving rectangular matrices.
 
 In our quick response~\cite{KM:TFAf} to the paper of Fazwi et al., we announced that we can find
 further schemes for $4\times4$ matrices using 47 multiplications in characteristic two, and that we
 can reduce the number of multiplications required for $5\times5$ matrices in characteristic two to~95.
 In the present paper, we introduce the method by which we found these schemes. 

 We define a graph whose vertices are correct matrix multiplication schemes and where there is an edge
 from one scheme to another if the second can be obtained from the first by some kind of
 transformation. We consider two transformations. One is called a flip and turns a given scheme to a
 different one with the same number of multiplications, and the other is called a reduction and turns
 a given scheme to one with a smaller number of multiplications.  The precise construction of this
 flip graph is given in Sect.~\ref{sec:flip-graph}. In Sect.~\ref{sec:2x2-3x3}, we illustrate
 (parts of) the flip graph for $2\times2$ and $3\times3$ matrices.

 In order to find better upper bounds for a specific format, we start from a known scheme, e.g.,
 the standard algorithm, and perform a random walk in the flip graph. Although reduction edges
 are much more rare than flip edges, it turned out that there are enough of them to reach interesting
 schemes with a reasonable amount of computation time. In particular, we were able to match the
 best known algorithms for all multiplication formats $n\times m$ times $m\times p$ with $n,m,p\leq 5$,
 and found better bounds in four cases. These results are reported in Sect.~\ref{sec:other-formats}. 

 \section{Matrix Multiplication}\label{sec:matr-mult}

 Let $K$ be a field, let $R$ be a $K$-Algebra and let $\A \in R^{n\times m}$, $\B\in R^{m\times
   p}$. Recall that the computation of the matrix product $\C=\A\B$ by a Strassen-like algorithm
 proceeds in two stages.  In the first stage we compute certain products $m_1, \ldots, m_r$ of
 linear combinations of entries of $\A$ and linear combinations of entries of $\B$.  In the second
 stage the entries of $\C$ are obtained as linear combinations of the~$m_i$.  An algorithm of this
 form is called a bilinear algorithm \cite{BCS:ACT}.

 For example, in the case $n=m=p=2$, if we let
 \begin{equation*}
   \A = \begin{pmatrix}
         a_{1,1}&a_{1,2}\\
         a_{2,1}&a_{2,2}\\
       \end{pmatrix},
   \quad
   \B = \begin{pmatrix}
         b_{1,1}&b_{1,2}\\
         b_{2,1}&b_{2,2}
       \end{pmatrix}
   \quad\text{and}\quad
   \C = \begin{pmatrix}
         c_{1,1}&c_{1,2}\\
         c_{2,1}&c_{2,2}
       \end{pmatrix},
 \end{equation*}             
 then Strassen's algorithm computes $\C$ in the following way:
 \begin{alignat*}{4}
   m_1 &= (a_{1,1} + a_{2,2}) (b_{1,1} + b_{2,2}) &\quad\smash{\raisebox{-6.4\baselineskip}{\rule\fboxrule{7\baselineskip}}}\quad c_{1,1} &= m_1 + m_4 - m_5 + m_7 \\
   m_2 &= (a_{2,1} + a_{2,2}) (b_{1,1}) & c_{1,2} &= m_3 + m_5 \\
   m_3 &= (a_{1,1}) (b_{1,2} - b_{2,2}) & c_{2,1} &= m_2 + m_4\\
   m_4 &= (a_{2,2}) (b_{2,1} - b_{1,1}) & c_{2,2} &= m_1 - m_2 + m_3 + m_6.\\
   m_5 &= (a_{1,1} + a_{1,2})(b_{2,2})\\
   m_6 &= (a_{2,1} - a_{1,1}) (b_{1,1}+ b_{1,2})\\
   m_7 &= (a_{1,2} - a_{2,2}) (b_{2,1} + b_{2,2})
 \end{alignat*}
 It is common and convenient to express matrix multiplication in the language of tensors.
 For $2\times2$ matrices, the matrix multiplication tensor is
 \[
 \mathcal{M}_{2,2,2}=\sum_{i,j,k=1}^2 E_{i,j}
 \otimes E_{j,k}\otimes E_{k,i}\in K^{2\times 2}\otimes K^{2\times 2}\otimes K^{2\times 2},
 \]
 where $E_{u,v}$ refers to the matrix having a $1$ at position $(u,v)$ and zeros elsewhere.
 Strassen's algorithm is based on the observation that this tensor can also be written as
 the sum of only seven tensors of the form $A\otimes B\otimes \varGamma$. Indeed, we have
 \begin{alignat*}{2}
   \mathcal{M}_{2,2,2} &= 
   (a_{1,1} + a_{2,2}) \otimes (b_{1,1} + b_{2,2}) \otimes (c_{1,1}+c_{2,2})\\
   &+(a_{2,1} + a_{2,2}) \otimes (b_{1,1}) \otimes (c_{1,2}-c_{2,2})\\
   &+(a_{1,1}) \otimes (b_{1,2} - b_{2,2}) \otimes (c_{2,1}+c_{2,2})\\
   &+(a_{2,2}) \otimes (b_{2,1} - b_{1,1}) \otimes (c_{1,1}+c_{1,2})\\
   &+(a_{1,1} + a_{1,2}) \otimes (b_{2,2}) \otimes (c_{2,1}-c_{1,1})\\
   &+(a_{2,1} - a_{1,1}) \otimes (b_{1,1}+ b_{1,2}) \otimes (c_{2,2})\\
   &+(a_{1,2} - a_{2,2}) \otimes (b_{2,1} + b_{2,2}) \otimes (c_{1,1}),
 \end{alignat*}
 where for better readability we write $a_{i,j},b_{i,j},c_{i,j}$ instead of~$E_{i,j}$.  Readers not
 comfortable with tensor products are welcome to understand $a_{i,j},b_{i,j},c_{i,j}$ as polynomial
 variables and $\otimes$ as multiplication. Note that each row in the expression above corresponds
 to one of the $m_i$'s from before. The first two factors encode their definitions and the third how
 they enter into the~$c_{j,i}$. Note also that the indices in the third factor are swapped in order
 to be consistent with the matrix multiplication literature, where this version is preferred because
 it makes certain equations more symmetric.

 The general definition is as follows.
 \begin{definition}
   Let $n,m,p\in\N$. The matrix multiplication tensor of the format $(n,m,p)$
   is defined as
   \[
   \Mnmp = \sum_{i,j,k = 1}^{n,m,p}a_{i,j}\otimes b_{j,k} \otimes c_{k,i}
   \in K^{n\times m}\otimes K^{m\times p}\otimes K^{p\times n},
   \]
   where $a_{u,v},b_{u,v},c_{u,v}$ refer to matrices of the respective format having a $1$ at position
   $(u,v)$ and zeros in all other positions. 

   Rank-one tensors are non-zero tensors that can be written as $A \otimes B \otimes \varGamma$ for
   certain matrices $A\in K^{n\times m}$, $B\in K^{m\times p}$, $\varGamma\in K^{p\times m}$.
   The rank of a tensor $\T$ is the smallest number~$r$ such that $\T$ can be written as a sum of $r$
   rank-one tensors.
   
   An $(n,m,p)$-matrix multiplication scheme is a finite set $S$ of rank-one tensors whose sum is~$\Mnmp$.
   We call $|S|$ the rank of the scheme.
 \end{definition}
 
 The sum of $k$ linearly independent rank-one tensors need not necessarily have tensor rank
 $k$.  Strassen's discovery amounts precisely to the observation that $\mathcal{M}_{2,2,2}$, which is
 defined as a sum of 8 linearly independent rank-one tensors, has only rank~7.  A more simple
 example is the equation
 \begin{alignat*}1
   &a_{1,1}\otimes b_{1,1}\otimes c_{1,1} \ + \
   a_{1,1}\otimes b_{1,2}\otimes c_{3,1} \ + \
   a_{1,1}\otimes (b_{1,1}+b_{1,2})\otimes c_{2,2}\\
   &=a_{1,1}\otimes b_{1,1}\otimes (c_{1,1}+c_{2,2}) \ + \ a_{1,1}\otimes b_{1,2}\otimes (c_{3,1}+c_{2,2}).
 \end{alignat*}
 If a matrix multiplication scheme contains the three rank-one tensors in the first line, we can
 replace them by the two rank-one tensors in the second line and thereby reduce the rank by one.
 This observation can be generalized as follows. 

  \begin{definition}\label{reducible}
   Let $n,m,p,r \in \N$ and let $S = \{A^{(i)}\otimes B^{(i)}\otimes \Gamma^{(i)}\mid i \in \range{r}\}$ be
   an $(n,m,p)$-matrix multiplication scheme.  We call $S$ \emph{reducible} if there is a nonempty
   set $I\subseteq \range{r}$ such that
   \begin{enumerate}
   \item $\dim_{K} \langle A^{(i)} \rangle_{i\in I} =1$ and
   \item $\{B^{(i)}\mid i\in I\}$ is linearly dependent over $K$,

   \end{enumerate}
   or analogously with $B,A$ or $A,\Gamma$ or $\Gamma,A$ or $B,\Gamma$ or $\Gamma,B$ in place of
   $A,B$.
 \end{definition}
 
 \begin{proposition}\label{reduction}
   Let $n,m,p,r \in \N$ and let $S$ be a reducible $(n,m,p)$-matrix multiplication scheme of rank~$r$.
   Then there exists an $(n,m,p)$-matrix multiplication scheme of rank~$r-1$.
 \end{proposition}
 \begin{proof}
   We prove the statement in the case that the conditions in Def.~\ref{reducible} hold for $A$ and
   $B$.  In the other cases the proof works analogously.  Since $\{B^{(i)}\mid i\in I\}$ is linearly
   dependent, there is a $t\in I$ such that $B^{(t)} = \sum_{i\in I\setminus\{t\}} \beta_i B^{(i)}$
   for some $\beta_i \in K$.  Moreover, there are $\alpha_{i}\in K$ such that
   $A^{(t)} = \alpha_{i}A^{(i)}$. Hence, we have
   \begin{alignat*}1
     &\sum_{i\in I}A^{(i)} \otimes B^{(i)} \otimes \Gamma^{(i)}\\
     &=\sum_{i\in I\setminus \{t\}}A^{(i)}\otimes B^{(i)}\otimes \Gamma^{(i)}+\alpha_{i}A^{(i)}\otimes
       \sum_{i\in I\setminus\{t\}} \beta_i B^{(i)} \otimes \Gamma^{(t)}\\
     &=\sum_{i\in I\setminus \{t\}}A^{(i)}\otimes B^{(i)}\otimes \Gamma^{(i)}+ \sum_{i\in
     I\setminus\{t\}} A^{(i)}\otimes  B^{(i)} \otimes \alpha_{i}\beta_{i}\Gamma^{(t)} \\
     &=\sum_{i\in I\setminus \{t\}}A^{(i)}\otimes B^{(i)}\otimes(\Gamma^{(i)}+\alpha_{i}\beta_i\Gamma^{(t)})
   \end{alignat*}
   Therefore,
   \begin{alignat*}{1}
     S'={}&\{A^{(i)}\otimes B^{(i)}\otimes \Gamma^{(i)} \mid i\in \range{r}\setminus I\} \\
        &{}\cup\{A^{(i)}\otimes B^{(i)}\otimes(\Gamma^{(i)}+\alpha_{i}\beta_i\Gamma^{(t)})\mid i\in I\setminus\{t\}\}
   \end{alignat*}
   is a multiplication scheme with rank $r-1$.
 \end{proof}
 We call a scheme~$S'$ constructed as above a \emph{reduction} of~$S$.

 Symmetries of the matrix multiplication tensor $\mathcal{M}_{n,m,p}$ give rise to various ways to
 transform a given matrix multiplication scheme into another one.  For example, because of the
 identity $(\A\B)^\top=\B^\top\A^\top$, exchanging each rank-one tensor $A\otimes B\otimes\varGamma$
 by $B^\top\otimes A^\top\otimes \Gamma^\top$ maps a correct scheme to another correct scheme. A
 correct scheme is also obtained if we replace every rank-one tensor $A\otimes B\otimes\varGamma$ by
 $B\otimes \Gamma\otimes A$. Finally, if $U\in K^{m\times m}$ is invertible, then the identity
 $\A\B=\A UU^{-1}\B$ implies that also replacing every rank-one tensor $A\otimes B\otimes\varGamma$
 by $AU\otimes U^{-1}B\otimes\varGamma$ maps a correct scheme to another one. These transformations
 generate the symmetry group of $\mathcal{M}_{n,m,p}$.  For more details on this group see
 \cite{dG:Ovoo,KM:ANFf}.  We can let the symmetry group act on the set of matrix multiplications and
 call two schemes equivalent if they belong to the same orbit.

 Def.~\ref{reducible} is compatible with the action of the symmetry group. For permutations,
 this is because the definition explicitly allows any two factors to take the roles of $A$ and~$B$,
 and linear dependence is preserved under transposition. Likewise, if some matrices $B^{(i)}$ ($i\in I$)
 are linearly dependent, then so are the matrices $VB^{(i)}W^{-1}$ ($i\in I$) for any invertible
 matrices~$V,W$. Therefore, we can extend Def.~\ref{reducible} from individual matrix multiplication
 schemes to equivalence classes.

 \section{The Flip Graph}\label{sec:flip-graph}

 As reducibility is preserved by the application of a symmetry, the only chance to get from an
 irreducible to a reducible scheme is by a transformation that, at least in general, does not
 preserve symmetry. Flips have this feature. The idea is to subtract something from one rank-one
 tensor and add it to one of the others. This can be done for any two rank-one tensors that
 share a common factor. For example, we have
 \begin{alignat*}1
   &A\otimes B\otimes \varGamma \ + \ A\otimes B'\otimes \varGamma' \\
   &= A\otimes (B+B')\otimes \varGamma \ + \ A\otimes B'\otimes (\varGamma' - \varGamma).
 \end{alignat*}
 In contrast to a symmetry transformation, this transformation only affects two elements of the
 multiplication scheme instead of all. Therefore, we can in general expect the resulting scheme to be
 non-equivalent to the original one.
 \begin{definition}\label{flip}
   Let $n,m,p,r\in \N$ and let $S,S'$ be $(n,m,p)$-matrix multiplication schemes of rank~$r$.  We call
   $S'$ a \emph{flip} of $S$ if there are
   \begin{enumerate}
   \item $T_{1} = A \otimes B \otimes \Gamma \in S$,
   \item $T_{2} = A \otimes B' \otimes \Gamma' \in S$ and
   \item $T \in \{A\otimes B\otimes \Gamma',A\otimes B'\otimes \Gamma\}$
   \end{enumerate}
   such that $(S\setminus\{T_{1},T_{2}\})\cup \{T_{1}+T,T_{2}-T\} = S'$.
   
   The definition is meant to apply analogously for any permutation of $A, B$ and~$\Gamma$.
 \end{definition}
 The union is always disjoint since we require $S$ and $S'$ to be of the same rank. Cases
 where $S'$ would have a smaller rank than $S$ need not be included because in these cases $S'$
 is a reduction of~$S$.
 Note that flips are reversible because if $S' = (S\setminus\{T_{1},T_{2}\})\cup \{T_{1}+T,T_{2}-T\}$,
 then $S = (S'\setminus\{T_{2}-T,T_{1}+T\})\cup \{(T_{2}-T)+T,(T_{1}+T)-T\}$.

 Flips are well-defined for equivalence classes of matrix multiplication schemes in the following
 sense: If $S'$ is a flip of $S$ and $g$ is an element of the symmetry group, 
 then $g(S)$ is a flip of $g(S')$. This is the case since Def.~\ref{flip} applies for arbitrary
 permuations of $A,B,\Gamma$ and because the symmetry transformations act linearly on the
 individual matrices.

 Def.~\ref{flip} introduces flips as rather specific ways of replacing two rank-one tensors in a
 given scheme by some other rank-one tensors. However, it turns out that if we only want to replace
 two rank-one tensors at a time, then there is not much more we can do. Thm.~\ref{th:completeness}
 makes this more precise. 

 \def\1{^{(1)}}\def\2{^{(2)}}\def\3{^{(3)}}\def\4{^{(4)}}

 \begin{lemma}\label{lemma:1}
   Let $U,V$ and $W$ be vector spaces.
   For $i=1,\dots,4$, let $A^{(i)}\in U, B^{(i)} \in V, \Gamma^{(i)}\in W$
   be such that
   \begin{align}\label{eq:1}
     \begin{split}
       &\t{1}+\t{2} \\
       &=\t{3}+\t{4}
     \end{split}
   \end{align}
   and
   \[
   \dim\langle A\1,A\2\rangle = \dim\langle B\1,B\2\rangle = \dim\langle \Gamma\1,\Gamma\2\rangle = 2.
   \]
   Then
   \begin{align*}
     &\bigl\{\t{1},\ \t{2}\bigr\}\\
     & = \bigl\{\t{3},\ \t{4}\bigr\}.
   \end{align*}
 \end{lemma}
 \begin{proof}
   Since $A\1$ and $A\2$, $B\1$ and $B\2$, and $\Gamma\1$ and $\Gamma\2$ are linearly independent,
   there are
   \begin{align*}
     &\hspace{-12ex}\alpha_{1},\alpha_{2},\alpha'_{1},\alpha'_{2}, \beta_{1},\beta_{2},\beta'_{1},
     \beta'_{2}, \gamma_{1},\gamma_{2},\gamma'_{1},\gamma'_{2}\in K,\hspace{-12ex}\\
     A,A'&\in U \setminus \langle A\1, A\2 \rangle,\\
     B,B'&\in V \setminus \langle B\1, B\2\rangle,\\
     \Gamma,\Gamma'&\in W \setminus \langle \Gamma\1, \Gamma\2 \rangle
   \end{align*}
   such that
   \begin{alignat*}{3}
     A\3 &= \alpha_{1} A\1 + \alpha_{2}A\2 + A, \qquad& A\4 &= \alpha'_{1} A\1 + \alpha'_{2}A\2 + A', \\  
     B\3 &= \beta_{1} B\1 + \beta_{2}B\2 + B, & B\4 &= \beta'_{1} B\1 + \beta'_{2}B\2 + B', \\  
     \Gamma\3 &= \gamma_{1} \Gamma\1 + \gamma_{2}\Gamma\2 + \Gamma, & \Gamma\4 &= \gamma'_{1} \Gamma\1 + \gamma'_{2}\Gamma\2 + \Gamma'.
   \end{alignat*}
   After plugging these in~(\ref{eq:1}) and equating coefficients we obtain the following system of
   equations:
   \begin{alignat*}{3}
     \alpha_{1}\beta_{1}\gamma_{1} + \alpha'_{1}\beta'_{1}\gamma'_{1} &= 1, \qquad & \alpha_{1}\beta_{1}\gamma_{2} + \alpha'_{1}\beta'_{1}\gamma'_{2} &= 0,\\
     \alpha_{1}\beta_{2}\gamma_{1} + \alpha'_{1}\beta'_{2}\gamma'_{1} &= 0, \qquad & \alpha_{1}\beta_{2}\gamma_{2} + \alpha'_{1}\beta'_{2}\gamma'_{2} &= 0,\\
     \alpha_{2}\beta_{1}\gamma_{1} + \alpha'_{2}\beta'_{1}\gamma'_{1} &= 0, \qquad & \alpha_{2}\beta_{1}\gamma_{2} + \alpha'_{2}\beta'_{1}\gamma'_{2} &= 0,\\
     \alpha_{2}\beta_{2}\gamma_{1} + \alpha'_{2}\beta'_{2}\gamma'_{1} &= 0, \qquad & \alpha_{2}\beta_{2}\gamma_{2} + \alpha'_{2}\beta'_{2}\gamma'_{2} &= 1.
   \end{alignat*}
   A straightforward computation confirms that these equations imply that either
   $\alpha_{1}\beta_{1}\gamma_{1} = \alpha'_{2}\beta'_{2}\gamma'_{2} = 1$ and all other variables
   are $0$ or $\alpha_{2}\beta_{2}\gamma_{2} = \alpha'_{1}\beta'_{1}\gamma'_{1} = 1$ and all other
   variables are $0$.  It follows further that $A,A',B,B',C,C'$ need to be zero, which completes the
   proof.
 \end{proof}

 \begin{lemma}\label{lemma:2}
   Let $U,V$ be vector spaces.
   Let $B\1,B\2,B\3,B\4\in U$ and $\Gamma\1,\Gamma\2,\Gamma\3,\Gamma\4\in V$
   be such that 
   \[
     B\1\otimes\Gamma\1+B\2\otimes\Gamma\2=B\3\otimes\Gamma\3+B\4\otimes\Gamma\4,
   \]
   $\dim\<B\1,B\2>=\dim\<\Gamma\1,\Gamma\2>=2$, and $\Gamma\3,\Gamma\4\neq 0$.
   
   Then $\<B\1,B\2>=\<B\3,B\4>$ and $\<\Gamma\1,\Gamma\2>=\<\Gamma\3,\Gamma\4>$.
 \end{lemma}
 \begin{proof}
   Since $\Gamma\1, \Gamma\2, \Gamma\3, \Gamma\4\neq 0$ the equality implies that $B\1,B\2,B\3,B\4$
   are linearly dependent.  Therefore, there exist $\mu,\nu,\lambda\in K$ such that $B\1 =
   \mu B\2 + \nu B\3 + \lambda B\4$.  It follows
   \begin{equation*}
     B\2 \otimes (\Gamma\2 + \mu \Gamma\1)
     =B\3 \otimes (\Gamma\3 -\nu \Gamma\1) +
     B\4 \otimes (\Gamma\4 -\lambda \Gamma\1).
   \end{equation*}
   
   Since $\Gamma\1$ and $\Gamma\2$ are linearly independent this implies
   $B\2 \in \langle B\3,B\4 \rangle$. For the same reason we have $B\1 \in \langle B\3,B\4 \rangle$
   and since $B\1$ and $B\2$ are linearly independent it follows $\<B\1,B\2>=\<B\3,B\4>$.
   The equality
   $\<\Gamma\1,\Gamma\2>=\<\Gamma\3,\Gamma\4>$ is shown by the same argument.
 \end{proof}
 
 \begin{theorem}\label{th:completeness} 
   Let $U,V$ and $W$ be vector spaces.
   For $i=1,\dots,4$, let $A^{(i)}\in U\setminus\{0\}, B^{(i)} \in V\setminus\{0\}, \Gamma^{(i)}\in W\setminus\{0\}$
   be such that
   \begin{align}\label{eq:2}
     \begin{split}
       &\t{1}+\t{2} \\
       &=\t{3}+\t{4},
     \end{split}
   \end{align}
   $\dim\<B\1,B\2>=\dim\<\Gamma\1,\Gamma\2>=2$, and
   \begin{align*}
     &\bigl\{\t{1},\ \t{2}\bigr\}\\
     &\neq\bigl\{\t{3},\ \t{4}\bigr\}.
   \end{align*}
   Then the following hold:
   \begin{enumerate}
   \item $\dim\<A\1,A\2,A\3,A\4> = 1$
   \item $\<B\1,B\2> = \<B\3,B\4>$
   \item $\<\Gamma\1,\Gamma\2> = \<\Gamma\3,\Gamma\4>$
   \end{enumerate}
 \end{theorem}

 \begin{proof}
   We start by showing the first claim.  If $A\1$ and $A\2$ are linearly independent, then it
   follows from Lemma~\ref{lemma:1} that
   \begin{align*}
     &\bigl\{\t{1},\ \t{2}\bigr\}\\
     & = \bigl\{\t{3},\ \t{4}\bigr\}.
   \end{align*}
   Thus, $A\1$ and $A\2$ must be linearly dependent.  So there exists $\alpha \in K$ such that
   $A\2 = \alpha A\1$.  It follows
   \begin{align*}
     &A\1 \otimes (B\1 \otimes \Gamma\1 + \alpha (B\2 \otimes \Gamma\2)) \\
     & = \t{3}+\t{4}.
   \end{align*}
   This can only be the case if $A\3$ and $A\4$ or $\bc{3}$ and $\bc{4}$ are
   linearly dependent.  If $\bc{3}$ and $\bc{4}$ are linearly dependent, then the right side
   of~(\ref{eq:2}) has rank one, so the left side also have rank $1$.  This would imply that
   either $\dim\<B\1,B\2> = 1$ or $\dim\<\Gamma\1,\Gamma\2>=1$. Therefore $A\3$ and $A\4$ are
   linearly dependent and so $\dim\<A\1,A\2,A\3,A\4> = 1$.
     
   Using that the $A^{(i)}$ are constant multiples of each other, and that
   $(\lambda A)\otimes B\otimes\Gamma=A\otimes(\lambda B)\otimes\Gamma$ for all $\lambda\in K$, we
   may as well assume $A\1=A\2=A\3=A\4$.  The equality~(\ref{eq:2}) then reduces to
   \begin{equation*}
     B\1\otimes\Gamma\1+B\2\otimes\Gamma\2=B\3\otimes\Gamma\3+B\4\otimes\Gamma\4.
   \end{equation*}
   
   The remaining claims follow directly from Lemma~\ref{lemma:2}.   
 \end{proof}

 Informally, Thm.~\ref{th:completeness} says that if we want to replace two rank-one tensors
 nontrivially by two others, then they must agree in one of the factors and this factor cannot
 change.  Additionally, the vector spaces generated by the other factors must stay the same.
 
 We now introduce the main concept of this paper.  
 
 \begin{definition}
   Let $n,m,p\in\set N$ and let $V$ be the set of all orbits of $(n,m,p)$-matrix multiplication
   schemes under the symmetry group and define
   \begin{align*}
     E_{1}&=\{(S,S')\mid S' \text{ is a flip of } S\}\\
     E_{2}&=\{(S,S')\mid S' \text{ is a reduction of } S\}.\\
   \end{align*}
   \begin{enumerate}
   \item The graph $G = (V,E_{1}\cup E_{2})$ is called the $(n,m,p)$-\emph{flip graph}.
     The edges in $E_1$ are called \emph{flips} and the edges in $E_2$ are called \emph{reductions.}
   \item For a given $r\in\set N$, the subgraph of $G$ consisting of all vertices of rank at most~$r$
     is called the $(m,n,p)$-\emph{flip graph} of rank at most~$r$.
   \item For a given $r\in\set N$, the set $\{S\in V:\rank(S)=r\}$ is called the $r$th \emph{level}
     of~$G$.
   \end{enumerate}
 \end{definition}

 Note that flips always connect vertices belonging to the same level, whereas a reduction always leads
 to a vertex belonging to a lower level. Also keep in mind that if there is a flip from $S$ to~$S'$,
 then there is one from $S'$ to~$S$. The flip graph may have loops, these correspond to flips that
 accidentally turn a certain scheme into an equivalent one.

 Since we are interested in schemes of low rank, we are interested in paths containing reductions,
 because these lead us into lower levels. We have not introduced any edges leading to higher levels,
 although that would be an easy thing to do. For example, given a scheme~$S$ containing a rank-one
 tensor $A\otimes B\otimes \Gamma$, we can replace this tensor by the two tensors
 $A\otimes B\otimes (\Gamma-\Gamma')$ and $A\otimes B\otimes\Gamma'$, for arbitrary
 $\Gamma'\in K^{p\times n}$.  The result is a scheme that admits a reduction to~$S$.  We could call this
 step a split of~$\Gamma$.  A split produces a correct matrix multiplication scheme as long as the
 original scheme does not already contain any of the newly added elements.  With this observation,
 we can show that the flip graph is connected.

 \begin{theorem}\label{thm:connected}
   For every $n,m,p\in\set N$, the $(n,m,p)$-flip graph is weakly
   connected, i.e., the undirected graph obtained from it by replacing every reduction by a
   bidirectional edge, is connected. 
 \end{theorem}
 \begin{proof}
   For any given scheme $S_0$, we construct a path to the standard algorithm in the underlying
   undirected graph. The first part of the path consists of reductions to an
   irreducible scheme $S_{1}$.  Then any two elements
   $A \otimes B \otimes \Gamma, A' \otimes B' \otimes \Gamma' \in S_{1}$ have the property that
   $A \otimes B$ and $A' \otimes B'$ are linearly independent.  This ensures that splits of $\Gamma$
   lead to pairwise distinct rank-one tensors.  Next we repeatedly split $\Gamma$ for every element
   to construct a scheme $S_{2}$ such that every element of $S_{2}$ can be written as
   $A \otimes B \otimes c_{i,j}$ for some $i,j\in \N$.

   For any two elements $A \otimes B \otimes \Gamma, A' \otimes B' \otimes \Gamma' \in S_{2}$ where
   $A$ and $A'$ are linearly dependent and $\Gamma$ and $\Gamma'$ are linearly dependent there is a
   reduction that combines these two elements.  We follow such reductions until we get a scheme
   $S_{3}$ where any two elements
   $A \otimes B \otimes \Gamma, A' \otimes B' \otimes \Gamma' \in S_{3}$ have the property that
   $A \otimes \Gamma$ and $A' \otimes \Gamma'$ are linearly independent. Then we can repeatedly
   split $B$ for every element to construct a scheme $S_{4}$ where every element has the form
   $A \otimes b_{i,j} \otimes c_{k,l}$.  We then use reductions to combine all elements of $S_{4}$
   with matching $b_{i,j}$ and $c_{k,l}$.  This way we get matrix multiplication scheme $S_{5}$
   which for all $i,j,k,l\in \N$ contains at most one element of the form
   $A \otimes b_{i,j} \otimes c_{k,l}$. Therefore, $S_{5}$ must be the standard algorithm.
 \end{proof}

 We are interested in paths in the flip graph that lead to schemes of low rank. Such paths are more
 likely to exist if there are many flips. Therefore, we select a ground field $K$ for which we can
 expect the number of flips to be large. As specified in Def.~\ref{flip} and justified in
 Thm.~\ref{th:completeness}, a flip is possible whenever a scheme contains two rank-one tensors
 sharing a common factor.  The chances for a common factor are higher if the field~$K$ is small,
 because the smaller the field, the smaller the number of possible factors. For this reason, we
 consider the ground field $K=\set Z_2$ in the following experiments.  For this ground field
 Thm.~\ref{th:completeness} implies that the only way to replace two rank-one tensors in a scheme is
 a flip.

 \begin{corollary}\label{cor:completeness}
   Let $r\in \N$ and let $S$ and $S'$ be two irreducible matrix multiplication schemes of rank $r$
   over $\Z_{2}$ that differ in exactly two elements.  Then $S'$ is a flip of $S$.
 \end{corollary}
 \begin{proof}
   Let $\t{1},\t{2}\in S$ and $\t{3},\t{4} \in S'$ be those elements.  Then we have
   \begin{align*}
       &\t{1}+\t{2} \\
       &=\t{3}+\t{4}
   \end{align*}
   and
   \begin{align*}
     &\bigl\{\t{1},\ \t{2}\bigr\}\\
     &\neq \bigl\{\t{3},\ \t{4}\bigr\}.
   \end{align*}
   Since $S$ is not reducible at least two of $\<A\1,A\2>$,$\<B\1,B\2>$ and $\<\Gamma\1,\Gamma\2>$
   have dimension two, say
   \[
     \dim\<B\1,B\2>=\dim\<\Gamma\1,\Gamma\2>=2.
   \]
   Then from Theorem~\ref{th:completeness} follows that $\dim\<A\1,A\2,A\3,A\4> = 1$ and therefore
   we also have $A\1=A\2=A\3=A\4$.  Moreover, $\<B\1,B\2> = \<B\3,B\4>$ and
   $\<\Gamma\1,\Gamma\2> = \<\Gamma\3,\Gamma\4>$.  Because of $K=\Z_{2}$, this can only be if
   $B\3 = B\1+B\2$ or $\Gamma\3 = \Gamma\1 + \Gamma\2$ and likewise for $B\4$ and $\Gamma\4$.  Thus,
   $S'$ is a flip of $S$.
 \end{proof}

 \section{$2\times 2$ matrices and $3\times 3$ matrices}\label{sec:2x2-3x3}

 \begin{figure*}
   \begin{center}
     \includegraphics{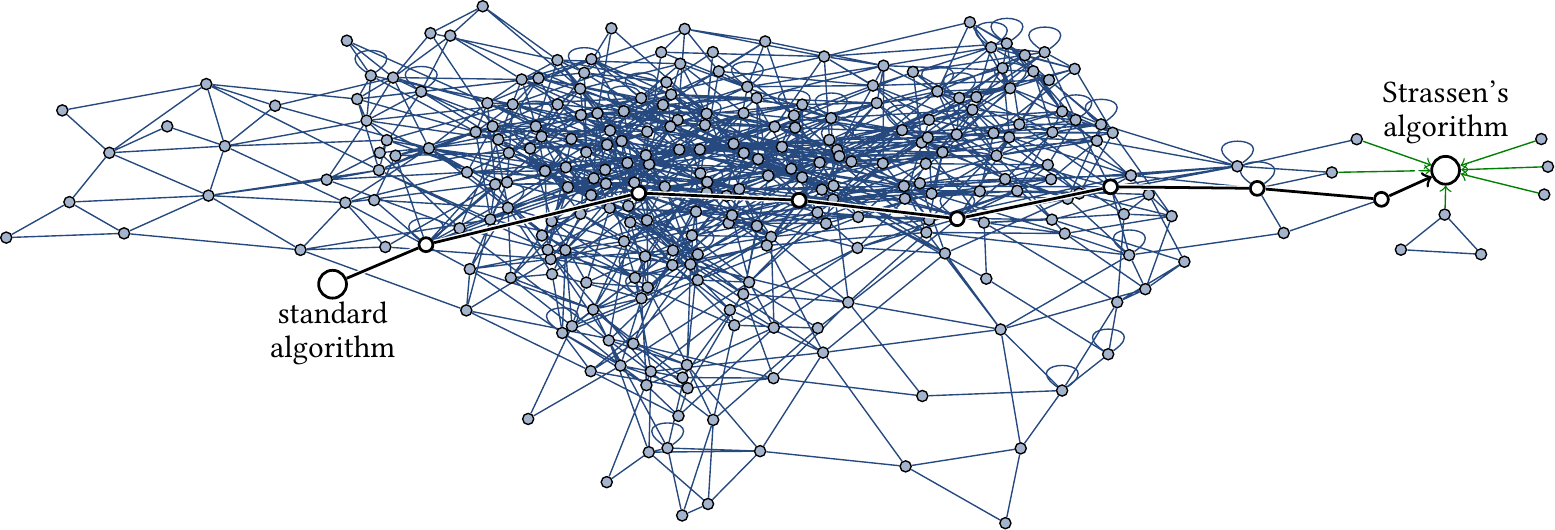}
   \end{center}
   \caption{In the $(2,2,2)$-flip graph of rank at most~8, this figure shows the component
     containing the standard algorithm.  Flips are depicted by undirected edges and reductions by
     directed edges.}
   \label{fig:2x2}
 \end{figure*}
  
 The $(2,2,2)$-flip graph of rank at most~8 for $K=\set Z_2$ is not too big. Fig.~\ref{fig:2x2}
 shows the connected component to which the standard algorithm belongs. It has 272 vertices, each
 representing the orbit of one multiplication scheme, and 1183 edges, 7 of which are reductions
 (shown by green arrows). The component also contains Strassen's algorithm. The distance between the
 standard algorithm and Strassen's algorithm is~8, a path is highlighted in the figure.
 Although the standard algorithm
 allows many flips, it only has one neighbor, because any two schemes obtained by a flip from the
 standard algorithm are equivalent. The diameter of the component is~12.

 Using SAT solvers as in~\cite{heule19a,HKS:Nwtm}, we have tried to find out whether there are other matrix multiplication
 schemes of rank at most~8. For $2\times2$ matrices, modern SAT solvers have no trouble generating many
 multiplication schemes in a short time. Strangely enough, while we found many solutions belonging to the
 connected component shown in Fig.~\ref{fig:2x2}, we only found exactly one solution (up to symmetries)
 that does not belong to this component:
 \begin{alignat*}1
   \mathcal{M}_{2,2,2}&=(a_{1,1} + a_{2,2})\otimes(b_{1,1} + b_{2,2})\otimes(c_{2,1})\\
   &+(a_{2,2})\otimes(b_{1,1} + b_{2,1})\otimes(c_{1,2} + c_{2,1} + c_{2,2})\\
   &+(a_{2,1} + a_{2,2})\otimes(b_{1,1})\otimes(c_{1,2} + c_{2,2})\\
   &+(a_{1,1})\otimes(b_{1,2} + b_{2,2})\otimes(c_{1,1})\\
   &+(a_{1,2})\otimes(b_{1,1} + b_{1,2} + b_{2,1} + b_{2,2})\otimes(c_{1,1} + c_{2,1} + c_{2,2})\\
   &+(a_{1,2} + a_{2,1})\otimes(b_{1,1} + b_{1,2})\otimes(c_{2,2})\\
   &+(a_{1,2} + a_{2,2})\otimes(b_{2,1} + b_{2,2})\otimes(c_{2,1} + c_{2,2})\\
   &+(a_{1,1} + a_{1,2})\otimes(b_{1,1} + b_{1,2} + b_{2,2})\otimes(c_{1,1} + c_{2,1}).
 \end{alignat*}
 This scheme has no neighbors and thus forms a connected component of its own. We do not know
 whether the $(2,2,2)$-flip graph of rank at most~8 has any further components.
 
 For $3\times3$-matrices and $K=\set Z_2$, the flip graph is so large that it is no longer possible
 to determine the entire component of the standard algorithm in the $(3,3,3)$-flip graph of rank at
 most~27. Again, and for the same reason as before, the standard algorithm itself has only one
 neighbor. At distance~2, we found 600 vertices, at distance~3 there are about 20000, and at
 distance~4 nearly 600000. None of them is reducible. Computing the whole neighborhood of distance~5
 is infeasible.

 With long random walks however, it is quite likely to encounter reducible vertices. We employ the
 following simple procedure to search for reductions.
 
 \begin{algorithm}\label{alg:path}
   Input: A matrix multiplication scheme $S$ and a limit $\ell$ for the path length.\\
   Output: A matrix multiplication scheme with rank decreased by one or $\bot$.
   \step10 if $S$ has no neighbours, return~$\bot$
   \step20 for $i=1,\dots,\ell$, do:
   \step31 if $S$ is reducible, then return a reduction of~$S$.
   \step41 if one of the neighbours of $S$ is reducible, then return a reduction of it.
   \step51 Set $S$ to a randomly selected neighbour of~$S$.
   \step60 return $\bot$
 \end{algorithm}

 An implementation of this procedure in C can explore paths of lengths $10^8$ within minutes and
 needs almost no memory.

 Starting from the standard algorithm we easily find schemes of rank~23, matching the record set by Laderman in
 1976~\cite{La:Anaf}, but we found no scheme of rank~22. Restricting the lengths of the random walks 
 to~$10^7$, more than 95\% of the walks reach a scheme of rank~23, and almost all a scheme of
 rank~24. Recall that along a random walk, the rank can only decrease but not increase. In
 Fig.~\ref{pathlength}, we show for 10000 random walks after how many steps they reach a scheme of a
 specific rank.  

 \begin{figure}
   \includegraphics[scale=0.7]{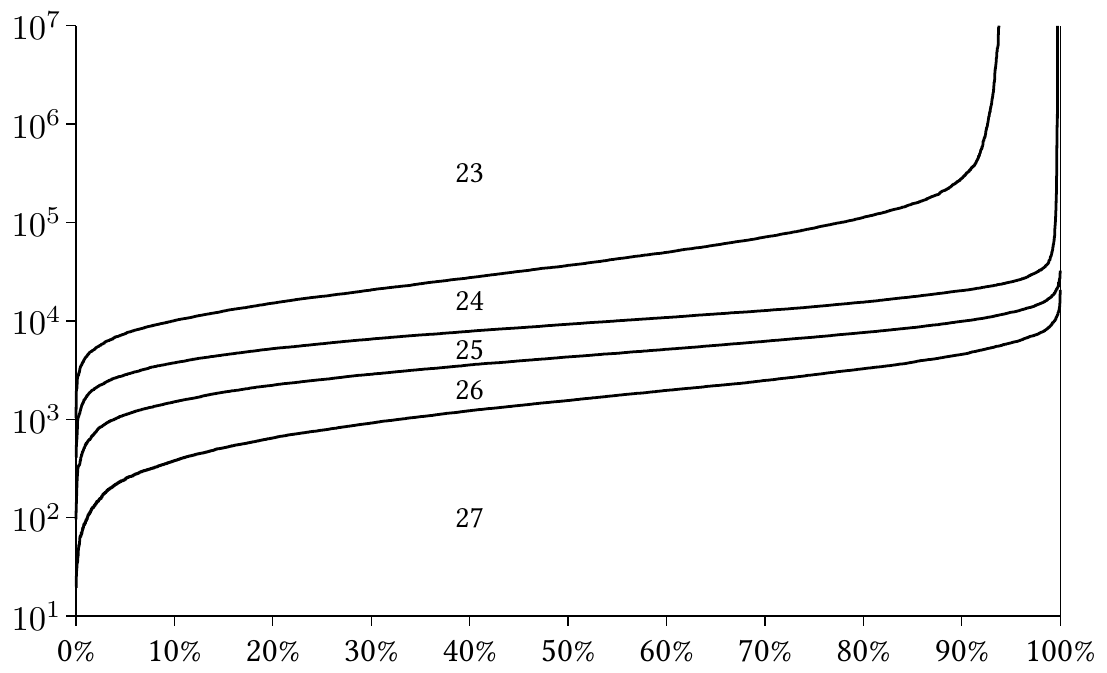}
   \caption{Sparsity of reduction steps. If a point $(x,y)$ belongs to a region labeled~$r$, then
     for $x$ random paths starting from the standard algorithms, the $y$th vertex has rank~$r$.}
   \label{pathlength} \end{figure}

 If we now focus on the flip graph of rank at most~23, it is feasible
 to determine for a given vertex the entire connected component to which it belongs. The schemes we reached
 by random walks from the standard algorithm turned out to belong to 584 different connected components with
 altogether 64061 vertices. The components are quite diverse with respect to size and symmetry; Fig.~\ref{fig:3x3}
 shows three examples.
 The component on the top has 681 vertices and has no automorphisms. Some connected components are isomorphic
 to each other, typically such components enjoy nontrivial symmetries. For example, the component shown at the
 bottom right of Fig.~\ref{fig:3x3} has an automorphism group of order 576 and appears 32 times.

 The small component shown on the left appears 39 times and has an automorphism group of order 8.
 The triangular structure in this component appears often in the flip graph.  It originates from the
 two possibilities to choose $T$ for a flip.  Two flips of a scheme that use the same rank-one
 tensors in the flip always are adjacent.  The square structure in the middle appears whenever there
 is an element that shares a different factor with each of two other elements of a scheme.

 There are also components consisting of a single vertex. Laderman's scheme is such an example.

 \begin{figure}
   \begin{center}
     \includegraphics{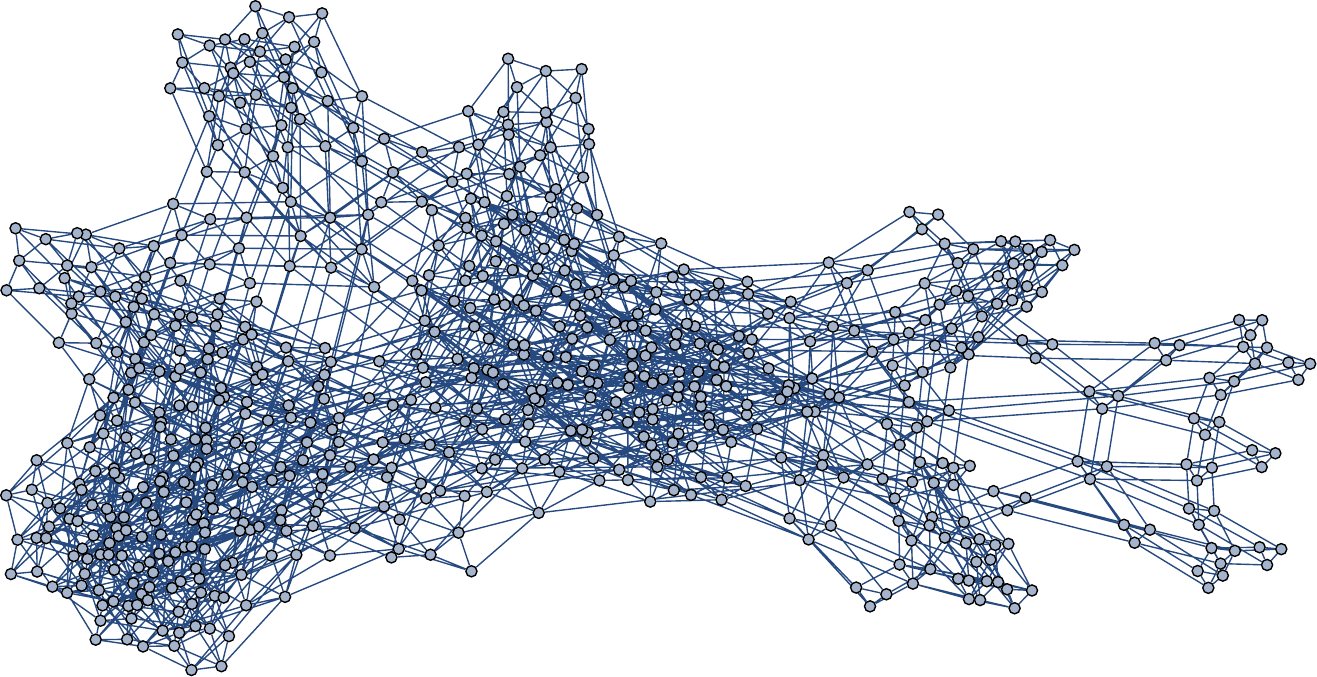}
     
     \includegraphics{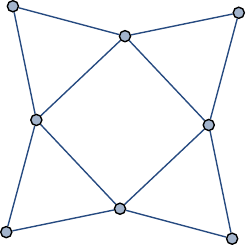}\hfil
     \includegraphics{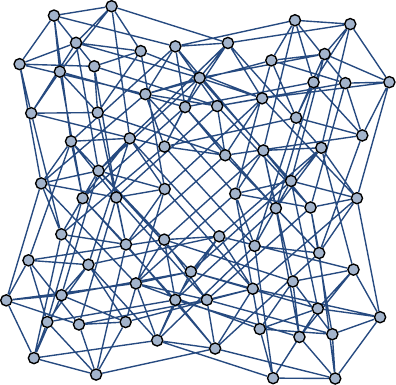}
   \end{center}
   \caption{Three example components of the $(3,3,3)$-flip graph of rank at most~23.}
   \label{fig:3x3}
 \end{figure}

 It seems that not all schemes of rank~23 can be reached from the standard algorithm, because the
 64000 solutions we were able to obtain by random paths starting from the standard algorithm do not
 include all the 17000 solutions found in~\cite{HKS:Nwtm} using SAT solving. Of course, given the size
 of the graph and the lengths of the paths, it is virtually impossible to check whether there really
 is no path or we were just not lucky enough to find it.
 


 \section{Other formats}\label{sec:other-formats}

 For larger matrix formats, it becomes harder to find random paths starting from the standard
 algorithm that go all the way down to a scheme of low rank. An adjusted search strategy that
 simultaneously considers many partial random paths was found to work more efficiently. In this
 variant, we maintain a pool of schemes of a certain rank~$r$ from which we randomly choose starting
 points, then do random walks starting from there until either a length limit or a scheme of rank $r-1$
 is encountered. In the latter case, the new scheme is saved. The procedure is repeated until a
 prescribed number of schemes of rank $r-1$ is reached. Then these schemes form the new pool of
 starting points and the method is repeated until the desired target rank is reached. 

 \begin{algorithm}\label{alg:search}
   Input: A set $P$ of schemes of a certain rank, a path length limit~$\ell$, a pool size limit~$s$, and a target rank~$r$\\
   Output: A set $Q$ of $s$ schemes of rank~$r$
   
   \step10 if $P$ consists of schemes of rank~$r$, return $P$.
   \step20 $Q=\emptyset$
   \step30 while $|Q|<s$ do:
   \step41 apply Alg.~\ref{alg:path} to a random element of~$P$ and~$\ell$.
   \step51 if Alg.~\ref{alg:path} returns a scheme, add it to~$Q$.
   \step60 call the algorithm recursively with $Q$ in place of~$P$.
 \end{algorithm}
 
 
 For our experiments, we used as $P$ the set containing only the standard algorithm, $\ell=10^6$ and
 $s=20000$.  With these settings, we were able to find schemes matching the best known rank bounds
 for all $(n,m,p)$ with $2\leq n,m,p\leq 5$, except for $(n,m,p)=(5,5,5)$. For the latter case,
 starting from the standard algorithm we only get down to rank 97 while Fawzi et
 al.~\cite{FBH+:Dfmm} discovered a scheme of rank 96 (valid mod~2). However, taking their scheme as starting
 point of a random walk, we discovered schemes of rank~95 within seconds. One of these schemes we
 announced in~\cite{KM:TFAf}. For $(n,m,p)=(4,4,5)$, they give a scheme of rank~63, improving the
 previous record by one, while we were able to find a scheme of rank~60 starting from the standard
 algorithm. Also this scheme is only valid mod~2.

 One of the remarkable outcomes of the recent work of Fawzi et al.~\cite{FBH+:Dfmm} is an apparent
 discrepancy of the rank depending on the characteristic of the ground field. Their scheme for
 $(n,m,p)=(4,4,4)$ of rank 47 as well as their scheme for $(n,m,p)=(5,5,5)$ of rank 96 are only
 valid in characteristic two and can be shown not to be the homomorphic image of a scheme for
 $K=\set Q$. As our search in the flip graph uses $K=\set Z_2$, the question is whether our schemes
 are also restricted to ground fields of characteristic two. To answer this question, we have applied
 Hensel lifting~\cite{vG:MCA} to the schemes we discovered.

 A set $S=\{ ((\alpha^{(\ell)}_{i,j}))\otimes((\beta^{(\ell)}_{j,k}))\otimes((\gamma^{(\ell)}_{k,i})):\ell=1,\dots,r\}$
 of rank-one tensors is a matrix multiplication scheme if and only if the cubic equations
 \[
   \sum_{\ell=1}^r \alpha^{(\ell)}_{i_1,i_2}\beta^{(\ell)}_{j_1,j_2}\gamma^{(\ell)}_{k_1,k_2} = \delta_{i_2,j_1}\delta_{j_2,k_1}\delta_{k_2,i_1}
 \]
 are satisfies for all $i_1,i_2,j_1,j_2,k_1,k_2$, where $\delta$ is the Kronecker delta function. These equations
 are known as the Brent equations~\cite{Br:Afmm}, and finding a matrix multiplication scheme is equivalent to solving
 these equations.

 Knowing a solution valid mod $2^s$ for some $s\in\set N$, we can view it as an approximation to
 order~$s$ of a solution valid in the $2$-adic integers, make an ansatz with undetermined
 coefficients for a refinement of this approximation to order $s+1$, plug this ansatz into the Brent
 equations, reduce mod $2^{s+1}$ and divide by~$2^s$.  This leads to a linear system over $\set Z_2$
 for the undetermined coefficients in the ansatz, which can be solved with linear algebra. If it has
 no solution, this proves that the approximation does not admit any refinement to order $s+1$. If it
 does have a solution, we pick one and proceed to refine. Once a decent approximation order is
 reached (we generously used $s=100$ although much less would have been sufficient in most cases),
 we can apply rational reconstruction~\cite{vG:MCA} to find a candidate solution with coefficients in
 $\set Q$ or even $\set Z$.  Whether the reconstruction was successful, i.e., whether the candidate
 solution over $\set Q$ or $\set Z$ is indeed a solution, can be checked easily by plugging it into
 the Brent equations.

 Proceeding as described above, we can confirm the mismatch for $(n,m,p)=(4,4,4)$ and
 $(n,m,p)=(5,5,5)$ observed by Fawzi et al.~\cite{FBH+:Dfmm}: none of our more than 100000 schemes
 of rank 47 for $(n,m,p)=(4,4,4)$ and none of our more than 30000 schemes of rank 95 for
 $(n,m,p)=(5,5,5)$ can be lifted from $\set Z_2$ to $\set Z_4$. However, we were able to lift a
 scheme of rank 97 for $(n,m,p)=(5,5,5)$ from $\set Z_2$ to $\set Z$, thereby breaking the record
 set by Smirnov and Sedoglavic~\cite{SS:TTRo} for this size. Moreover, while none of our schemes of
 rank 60 for $(n,m,p)=(4,4,5)$ could be lifted, we were able to lift some of the schemes of rank~62,
 thereby breaking the record set by Fawzi et al.~\cite{FBH+:Dfmm} for this size. For all other
 formats, we found no improvements but were able to match the best known bounds on the rank. An
 overview over the current state of affairs is given in Table~\ref{tab:ranks}.

 \begin{table}
   \begin{center}
     \begin{tabular}{c|c|c|c}
       $(n,m,p)$ & $K$ & previously best & our rank \\
             &     & known rank & \\\hline
       $(2,2,2)$ & any & 7 \cite{St:Gein} & 7 \\
       $(2,2,3)$ & any & 11 \cite{HK:OMtN} & 11 \\ 
       $(2,2,4)$ & any & 14 \cite{HK:OMtN} & 14 \\ 
       $(2,3,3)$ & any & 15 \cite{HK:OMtN} & 15 \\ 
       $(2,2,5)$ & any & 18 \cite{HK:OMtN} & 18 \\ 
       $(2,3,4)$ & any & 20 \cite{HK:OMtN} & 20 \\
       $(3,3,3)$ & any & 23 \cite{La:Anaf} & 23\\
       $(2,3,5)$ & any & 25 \cite{HK:OMtN} & 25 \\ 
       $(2,4,4)$ & any & 26 \cite{HK:OMtN} & 26 \\ 
       $(3,3,4)$ & any & 29 \cite{Sm:Tbca} & 29 \\ 
       $(2,4,5)$ & any & 33 \cite{HK:OMtN} & 33 \\
       $(3,3,5)$ & any & 36 \cite{Sm:Tbca} & 36 \\
       $(3,4,4)$ & any & 38 \cite{Sm:Tbca} & 38 \\
       $(2,5,5)$ & any & 40 \cite{HK:OMtN} & 40 \\
       $(3,4,5)$ & any & 47 \cite{FBH+:Dfmm} & 47 \\ 
       $(4,4,4)$ & $\set Z_2$ & 47 \cite{FBH+:Dfmm} & 47 \\ 
       $(4,4,4)$ & any & 49 \cite{St:Gein} & 49 \\
       $(3,5,5)$ & any & 58 \cite{SS:TTRo} & 58\\ 
       $(4,4,5)$ & $\set Z_2$ & 63 \cite{FBH+:Dfmm} & \textbf{60} \\
       $(4,4,5)$ & any & 63 \cite{FBH+:Dfmm} & \textbf{62} \\
       $(4,5,5)$ & any & 76 \cite{FBH+:Dfmm} & 76 \\
       $(5,5,5)$ & $\set Z_2$ & 96 \cite{FBH+:Dfmm} & \textbf{95} \\
       $(5,5,5)$ & any & 98 \cite{SS:TTRo} & \textbf{97}
     \end{tabular}
   \end{center}
   \caption{Comparison between best known rank and the rank we found}
   \label{tab:ranks}
 \end{table}

 One scheme for each format and the implementation of the search procedure are available at
 \begin{center}
   \url{https://github.com/jakobmoosbauer/flips.git}.
 \end{center}
 The other schemes are available upon request.
 
 \section{Open questions}

 The flip graph offers a new explanation for the existence of Strassen's algorithm and has led us
 to improved matrix multiplication schemes for some formats. We believe that the flip graph is
 interesting in its own right and deserves to be better understood.
  
 For example, is not clear how well-suited the standard algorithm is as starting point for the
 search procedure. Thm.~\ref{thm:connected} states that if we allow to use reductions backwards, then
 every algorithm is reachable from the standard algorithm. However, we found a $(2,2,2)$-matrix
 multiplication scheme of rank~8 that is not connected to the standard algorithm by a path that uses
 only vertices of rank~8.

 \begin{question}
   For $n,m,p \in \N$, is there a rank $r$ such that all vertices in level~$r$ of the
   $(n,m,p)$-flip graph are reachable from the standard algorithm?
 \end{question}

 If there are schemes of low rank which can not be reached from the standard algorithm, they might
 become reachable if we add additional edges to the graph.  In Corollary~\ref{cor:completeness} we show
 that, at least over $\Z_{2}$, the only way to replace exactly two rank-one tensors in one step is a flip.

 \begin{question}
   Under which condition can we replace more than two rows of a matrix multiplication scheme at the
   same time, such that the resulting scheme is not necessarily reachable by a sequence of flips?
 \end{question}

 Another way to add edges to the graph would be to add for every reduction also the reverse edge.
 However, this would create a lot of additional edges to higher levels. It is not clear at which
 points in the search procedure one should go to a higher level and which of these edges to use.

 \begin{question}
   How can we utilize edges leading to a higher level in the search procedure?
 \end{question}

 The search procedure would also benefit if we could determine whether two vertices belong to the
 same connected component in the current level.  This would allow us to restrict the pool of schemes
 in Algorithm~\ref{alg:search} such that there are not too many vertices in the same component and
 thus the search potentially covers a larger part of the graph.

 \begin{question}
   Given two matrix mulatiplication schemes of the same format and rank, is there an efficient way
   to determine whether they are connected within one level of the flip graph?
 \end{question}

 More generally, in order to search for matrix multiplication schemes of low rank, there might be
 a better way than following random paths in the graph.

 \begin{question}
   Given a matrix multiplication scheme~$S$, is there a systematic way to find reduction
   edges that can be reached from~$S$?   
 \end{question} 

 Finally, we observed that many of the connected components in the $(3,3,3)$-flip graph of rank at most~23
 are highly symmetric. Understanding these symmetries would help understanding the structure of the flip graph
 and might be useful in the search procedure.

 \begin{question}
   What is the significance of the high symmetry in some of the large components in the $(3,3,3)$-flip
   graph of rank at most~23?
 \end{question}

 \par\medskip\noindent\textbf{Acknowledgement.}
 We thank Martina Seidl for offering some of her computing power for conducting the experiments reported in
 this paper and her student Max Heisinger for valuable technical support with their system. 
 
 \bibliographystyle{plain}
 \bibliography{main}
 
\end{document}